\documentclass[10pt,journal,compsoc]{IEEEtran}
\ifCLASSOPTIONcompsoc
  \usepackage[nocompress]{cite}
\else
  \usepackage{cite}
\fi
\usepackage{xcolor}
\usepackage{soul}
\usepackage{color}
\usepackage{graphicx}
\usepackage[acronym,shortcuts]{glossaries}
\usepackage{caption} 
\usepackage{amsmath, mathtools}
\usepackage{amsfonts,amsthm,bm}
\usepackage{amssymb}
\usepackage{balance}
\usepackage{subfigure}
\usepackage{gensymb}
\newtheorem{theorem}{Theorem}
\newtheorem{lemma}[theorem]{Lemma}
\usepackage{amssymb}
\usepackage{pifont}%
\usepackage[linesnumbered,ruled]{algorithm2e}
\usepackage{algpseudocode}
\newcommand{\cmark}{\ding{51}}%
\newcommand{\xmark}{\ding{55}}%

\newacronym{ber}{BER}{Bit Error Rate}
\newacronym{crp}{CRP}{Challenge-Response Pair}
\newacronym{dram}{DRAM}{Dynamic Random Access Memory}
\newacronym{ef}{EF}{Entropy Features}
\newacronym{hd}{HD}{Hamming Distance}
\newacronym{hs}{HS}{Helper Stream}
\newacronym{iot}{IoT}{Internt of Thing}
\newacronym{nvm}{NVM}{Non-Volatile Memory}
\newacronym{puf}{PUF}{Physical Unclonable Function}

\begin{document}

\title{EPUF: A Novel Scheme Based on Entropy Features of Latency-based DRAM PUFs Providing Lightweight Authentication in IoT Networks}

\author{Fatemeh~Najafi, Masoud~Kaveh, Mohammad~Reza~Mosavi, Alessandro~Brighente, and~Mauro~Conti,~\IEEEmembership{Fellow,~IEEE,}
\IEEEcompsocitemizethanks{\IEEEcompsocthanksitem F. Najafi and M. R. Mosavi are with the Department of Electrical Engineering, Iran University of Science and Technology, Tehran, Iran.%
\IEEEcompsocthanksitem M. Kaveh is with the Department of Information and Communication Engineering, Aalto University, Espoo, Finland.%
\IEEEcompsocthanksitem A. Brighente and M. Conti are with the Department of Mathematics and HIT Reasearch Center, University of Padova, Padova, Italy.}
}


\IEEEtitleabstractindextext{%
\begin{abstract}
\acp{puf} are hardware-oriented primitives that exploit manufacturing variations to generate a unique identity for a physical system. Recent advancements showed how \ac{dram} can be exploited to implement \ac{puf}s. \ac{dram} \ac{puf}s require no additional circuits for \ac{puf} operations and can be used in most of the applications with resource-constrained nodes such as \ac{iot} networks. However, the existing \ac{dram} \ac{puf} solutions either require to interrupt other functions in the host system, or provide unreliable responses due to their sensitiveness to the environmental conditions. 

In this paper, we propose EPUF, a novel strategy to extract random and unique features from \ac{dram} cells to generate reliable \ac{puf} responses. In particular, we use the bitmap images of the binary \ac{dram} values and their entropy features. We show via real device experiments that EPUF is approximately $1.7$ times faster than other state of the art solutions, achieves $100\%$ reliability, generates features with $47.79\%$ uniqueness, and supports a large set of \acp{crp} that leads to new potentials for \ac{dram} \ac{puf}-based authentication. We also propose a lightweight authentication protocol based on EPUF, which not only provides far better security guarantees but also outperforms the state-of-the-art in terms of communication overhead and computational cost.
\end{abstract}

}
\maketitle
\IEEEdisplaynontitleabstractindextext

\IEEEpeerreviewmaketitle

\IEEEraisesectionheading{\section{Introduction}}

\IEEEPARstart{P}{hysical} Unclonable Functions (PUFs) are inherently secure blocks that serve as information system for key generation, identification, and authentication purposes. Traditional cryptography methods need to store the secret keys in \ac{nvm} devices, which not only makes the systems vulnerable to physical attacks, but also incurs in significant costs. \ac{puf}-based mechanisms remove the need for storing secrets on \acp{nvm}, preventing the aforementioned attacks and reducing hardware costs. \ac{puf}s exploit physical features that are randomly generated during to the manufacturing process variations to generate unique values~\cite{1}. Dynamic  Random  Access  Memory (\ac{dram})-based \ac{puf}s have recently become popular security primitives, as they provide cost-efficient and functional security services. Furthermore, \ac{dram}s are major components of most modern electronic systems~\cite{2}. 

Existing \ac{dram} \ac{puf}s present the following limitations. \ac{dram} \ac{puf} mechanisms using power-up values require power cycles and significant latency to characterize unbiased \ac{dram} cells. Furthermore, they need a selection algorithm to extract stable cells~\cite{3}. Retention-based \ac{dram} \ac{puf}s need a long period of time (order of minutes) to generate sufficient failures. They also exploit complex error correction techniques that cause technical constraints and overheads~\cite{4}, \cite{5}. In \ac{dram} organization, there are multiple defined time limits for the proper scheduling of different operations. Modifying these parameters can cause the failure of the normal \ac{dram} activity and leakage of \ac{dram} data. The pattern of these failures is completely random and unpredictable, providing an inherent source for entropy. \ac{dram} latency \ac{puf}s are fast run-time accessible techniques that can benefit from the property to employ a \ac{dram} as a security primitive. These solutions target the reduction of the timing parameters and do not require power or refresh cycles. In~\cite{6}, the authors configured the \ac{dram} \ac{puf} by changing the time required for a complete activation process. Although this method needs a much shorter evaluation time, it requires filtering mechanisms to extract the cells with a high failure probability and enhance the \ac{puf} reliability. This approach is based on repeating the read operation multiple times and selecting the more stable bits to form the final responses. However, the filtering mechanism increases the evaluation time and causes hardware overheads. Another approach based on timing parameters has been recently proposed in~\cite{7}. This solution employs a collection of algorithms to extract robust responses in lower query time at the expense of increased memory overheads. PreLat \ac{puf} is also a latency-based \ac{puf}~\cite{8}, which uses the erroneous data caused by reduced pre-charge time to generate device signatures. In this structure, due to the deterministic behavior of \ac{dram} cells, the raw response may not have sufficient randomness and uniqueness. To address this issue, PreLat \ac{puf} selects the independent and suitable \ac{dram} cells to produce robust responses during a pre-selection algorithm. The second algorithm is employed to extract the address of qualified cells to form a challenge and access their contents during response generation. This technique leads to time and power overheads and limits the Challenge Response Pair (\ac{crp}) space of \ac{dram} \ac{puf} to independent cells.
In~\cite{9}, Orasa et al., proposed Dataplant as a new \ac{dram}-based security primitive, which makes changes in \ac{dram} timing signals to organize a more reliable and cost effective \ac{puf} mechanism. This strategy includes US-Dataplant that exploits Sense Amplifiers (SAs) to extract the random values without \ac{dram} data destruction, and UC-Dataplant, which generates unpredictable values by setting the voltage of \ac{dram} cell to precharge value and reading the content of it after activation process. Dataplant mechanism can improve the \ac{puf} performance in several cases, especially in reliability of responses. However, the implementation process of Dataplant mechanisms presents some limitations. In fact, it requires direct access to hardwired control logics in \ac{dram} to make some additional modifications in control signals. Therefore, it is not practical to use Dataplant in commodity \ac{dram}s.

In summary, the prior \ac{dram} \ac{puf}s need power cycles or considerable latency for \ac{puf} operation in most techniques may interrupt or interfere in other operations of the system, limiting the practicality of \ac{dram} \ac{puf}s in real-time applications. Furthermore, most \ac{dram} \ac{puf}s are highly sensitive to internal/external noises and ambient conditions such as temperature variation, which cause noisy and unreliable responses. In the existing \ac{dram} \ac{puf}s, this important challenge is addressed by exploiting pre-selection algorithms, filtering mechanisms, and post-processing techniques like helper data algorithms. However, all these methods contribute to significant additional hardware/computational overheads. 

In this paper, we propose a low-cost strategy as a new security primitive to extract unique and unpredictable features using entropy vectors of \ac{dram} bitmaps. We propose EPUF on the basis of existing fast-runtime latency-based \ac{dram} \ac{puf}s. This approach completely removes the need for extra power cycles, long waiting time, pre-selection algorithms, filtering mechanisms, and hardware modifications. Thanks to our method, we significantly improve the robustness of generated responses and minimize the implementation costs. Therefore, it can be readily used in commodity \ac{dram}s and Internet of Things (\ac{iot}) applications to provide security services. 
The contributions of our paper can be summarized as follows.
\begin{enumerate}
    \item We propose to use entropy features as a new approach to extract true random and highly reliable \ac{puf} responses from the raw \ac{dram} bits.
    \item We propose EPUF, a novel \ac{dram} \ac{puf} that can considerably address the drawbacks of the existing \ac{dram} \ac{puf}s. EPUF eliminates the need for post processing on the \ac{puf} responses and supports a large set of \acp{crp} (strong \ac{dram} \ac{puf}), which can be used as an intrinsic security primitive providing key generation and authentication purposes. 
    \item We consider an efficient methodology to implement the proposed EPUF with the capability of execution on low-cost microcontrollers. 
    \item We evaluate the characteristics of EPUF on real hardware, and present a performance comparison between EPUF and some of the state-of-the-art \ac{dram} \ac{puf}s. We show that EPUF is at least $1.7$ times faster than other solutions, achieves $100\%$ reliability, and $47.79\%$ response uniqueness.
    \item We propose a provably-secure authentication protocol based on EPUF for \ac{iot} networks, that outperforms the previous \ac{puf}-based schemes in terms of provided security features and lightweight design. 
\end{enumerate}

The remainder of this paper is organized as follows. Section~II presents the proposed EPUF. Section~III presents the implementation results and analysis. Section~IV demonstrates EPUF-based key generation. Section~V proposes a lightweight and provably-secure authentication protocol based on EPUF with in-depth security and performance analysis, and finally, Section~VI concludes the paper. 

\section{EPUF: A Novel DRAM PUF}
In this section we provide an overview of EPUF, our proposal for a novel \ac{dram} \ac{puf}. We first discuss in Section~\ref{sec:entropy fearures} how to compute entropy features to generate robust responses. Then, we discuss in Section~\ref{sec:implem} how EPUF exploits entropy features to generate unique responses.

\subsection{Entropy Features}\label{sec:entropy fearures}
\ac{puf}s are information systems that allow access to random information on the basis of their hardware characteristics. \ac{dram} \ac{puf} mechanisms provide random, unpredictable but still reproducible data using bit failures caused by intercepting normal \ac{dram} operations. During these mechanisms, however, only a small percentage of \ac{dram} cells fail and a large number of memory blocks is needed to detect a sufficient number of failures to generate responses. In addition, the procedure for selecting or filtering robust cells, as well as determining their location in the \ac{dram} array causes significant space, hardware, and computational costs. In this paper, we propose entropy features as a new primitive based on entropy values of raw \ac{dram} content and enable the generation of unpredictable bit-streams as reliable responses. Our proposed process to generate entropy features consists of three steps.

\textbf{Bitmap Image Generation.} In this phase, some known data is written to the \ac{dram} block with a specified starting address and block size. Then, the \ac{dram} data is read back when the timing parameters have been changed considering the used \ac{dram} \ac{puf} mechanism. The number of bitmap lines is given to the bitmap generator module as an initialization parameter. Then, bytes of output \ac{dram} binary data are scanned and converted to integer values to form bitmap pixels. 

\textbf{Entropy Module.} Entropy is a value that indicates how an irregular value appears. Bit failures in \ac{dram} \ac{puf} outputs occur in random, unpredictable, and \ac{dram} block-specific patterns that cause different and unpredictable entropy values for multiple bitmap lines. Therefore, entropy arrays for each \ac{dram} data-based bitmap are also unpredictable and unique, offering a new approach to \ac{puf} characterization. When the bitmap image is generated, the entropy array of the bitmap is calculated using the color values for each bitmap line. For example, in the case of a bitmap with $128$ pixels in each row, the entropy value of $128$ pixels is represented as one element in the entropy array. The length of the array is equal to the number of lines on the bitmap.  
The entropy value for each line is calculated using \eqref{eq:entropy}, where $E_j$ is the entropy of line $j$ and $p_i$ is the probability of  appearance of value $i$. The range of $i$ is set from $0$ to $255$, which shows the possible integer values for a byte. This phase including entropy calculations is executed by the entropy module, where entropy is computed as  
\begin{equation}\label{eq:entropy}
    E_j = -\sum_{i=0}^{i=255} p_i \log_2 p_i.
\end{equation}

\textbf{Robust Response Generation.} The elements of entropy array are floating point numbers due to the use of probability values and the logarithm function. In this phase, the floating point numbers are converted to fixed point to generate a response in the binary format. The number of fractional bits and the precision of the conversion can be determined considering the similarity of multiple measurements and the robustness of the \ac{dram} data read in several iterations. In this method, we analyze the similarities of entropy values belonging to different iterations. Then, we determine the best precision in which the final responses include only the shared points of entropy arrays. This parameter is fixed to the best value during the \ac{puf} characterization phase. We present our method to generate the \ac{ef} of a bitmap of a given memory segment size in \ref{alg:entFeat}. This algorithm entails the general procedure of extracting entropy characteristics of raw \ac{dram} responses that will be further processed to extract final robust responses. The algorithm receives as input the number $R$ of bitmap rows, the number $C$ of bitmap columns, and the number $B$ of bytes in each row.

\SetKwInput{Input}{Input}
\SetKwInput{Output}{Output}
\SetKw{Break}{break}
\begin{algorithm}[!h]
    
    \labelformat{algocf}{Algorithm\,#1}

    \Input{Address and Size of the PUF segment,the Input pattern, $R$, $C$, $B$.
}
    \Output{\ac{ef}}
    \tcc{initialize EF as empty array}
    EF $=[]$\;
    \tcc{1D array of size 256 containing the times each integer occurs in a row }
    $int_{\rm count} =[]$\; 
    \tcc{probabilities of appearances}
    $p=0$\;
    \For{$i=1$ \KwTo $R$}{
        \For{$j=1$ \KwTo $C$}{
        data = read\_memory(i,j)\;
        int\_data = binary\_to\_integer(data)\;
        int\_count(int\_data) ++\;
        }
        $p$= int\_count $/B$\;
        temp = $\sum_{k=0}^{255} p(k) \times \log p(k)$\;
        $EF(i)$ = integer\_to\_binary(temp)\;
    }
    \Return $EF$\;
    
    \caption{Generating entropy features.}
    \label{alg:entFeat}
\end{algorithm}

\subsection{Implementing EPUF Using Entropy Features}\label{sec:implem}
We propose EPUF, a solution based on entropy features that allows for the effective use of \ac{dram} segments and improves the \ac{crp} space. EPUF eliminates the effects of internal/external noises, resulting in higher reliability. Furthermore, EPUF can be implemented in commodity systems without any extra hardware, allowing for low-cost and high-performance authentication mechanisms. We discuss the process used for \ac{puf} characterization as well as the proposed mechanisms for key generation and authentication in following subsections. 

\textbf{Initial parameters.} To configure the EPUF, it is important to determine initial values for EPUF parameters, including the best EF precision, the number of bitmap lines, and the size of the \ac{dram} array. To find the best precision value to extract robust and reliable EFs and, finally, to provide reliable responses, we propose \ref{alg:initPar}. The algorithm receives as input the address and size of the \ac{puf} segment, the input pattern and $R$, and $C$, which represent respectively the number of rows and columns of the memory.
\SetKwInput{Input}{Input}
\SetKwInput{Output}{Output}
\SetKw{Break}{break}
\begin{algorithm}[!h]
    
    \labelformat{algocf}{Algorithm\,#1}

    \Input{Address and Size of the PUF segment,the Input pattern, number of reads $\Omega$.
}
    \Output{dmax = precision of fixed point binary EFs}
    dmax $=0$\;
    \For{$k=1$ \KwTo $\Omega$}{
    Calculate $EF_k$ using Algorithm 1\;
    }
    \For{$i=1$ \KwTo $\Omega$}{
        \For{$j=2$ \KwTo $\Omega$}{
            \If{$j>i$}{
            temp$ = \max \left (|EF_j|-|EF_i|\right)$\;
            }
            \If{temp $>$ dmax}{
            dmax $=$ temp\;
            }
        } 
    }    
    \Return dmax\;
    
    \caption{Setting initial parameters.}
    \label{alg:initPar}
\end{algorithm}
In \ref{alg:initPar}, raw \ac{dram} values are read multiple times, and the proper precision for EF calculations is determined considering the maximum differences between multiple iterations. The number of bitmap lines is determined by the length of the required \ac{puf} response and the precision of EFs. The \ac{dram} \ac{puf} segment is set to an efficient size considering the number of bitmap lines and the available \ac{dram} space, which is sufficient to extract the \ac{puf} responses without causing memory overheads. Finally, the bitmap columns are set by dividing the segment size by the number of lines. 

\textbf{Filtering Unreliable Features.} EPUF and EF-based responses provide a high degree of robustness and reliability by inducing the effects of internal/external noises. However, the entropy values of some memory cells in the \ac{puf} segments may include some differences under very different operating conditions. These dissimilarities happen due to the high sensitivity of the \ac{dram} cells to environmental conditions causing unreliability in the final responses. To address this problem, we propose \ref{alg:helpStream}, where each segment is evaluated under varying conditions, and the sensitive points that represent the entropy feature of most unstable \ac{dram} cells are extracted. Then, the \ac{puf} segment is filtered out of these features and a \ac{hs} is produced containing the stable features. \ac{hs} is sent to EPUF to filter unsuitable features from the final responses using a XNOR module. Additionally, it only contains the location of qualified entropy features and contains no information about \ac{dram}-based EPUF. Therefore, \ac{hs} is transferred without cryptography requirements, along with the input challenge to help the final responses achieve $100\%$ reliability.   

\section{EPUF Implementation and Analysis}
In this section, we present the EPUF implementation details, including the hardware architecture. We discuss the key metrics to evaluate our proposed method based on the experimental \ac{dram} data for \ac{dram} chips from different manufacturers. Finally, we demonstrate the EPUF performance analysis and compare it with other latency-based approaches. 

\subsection{Experimental Setup and Hardware Design}
 In this work, we exploit a platform consisting of Spartan6 FPGA and DDR3 chip for EPUF implementations.
In order to evaluate the reliability of EPUF and analyze its sensitivity to ambient conditions, we run the experimental tests under different temperatures ($25^{\circ}-55^{\circ}$ C). The hardware design used to implement the EPUF is summarized in Fig.~\ref{fig:arch}. 
\begin{figure}[!h]
    \centering
    \includegraphics[width=.9\columnwidth]{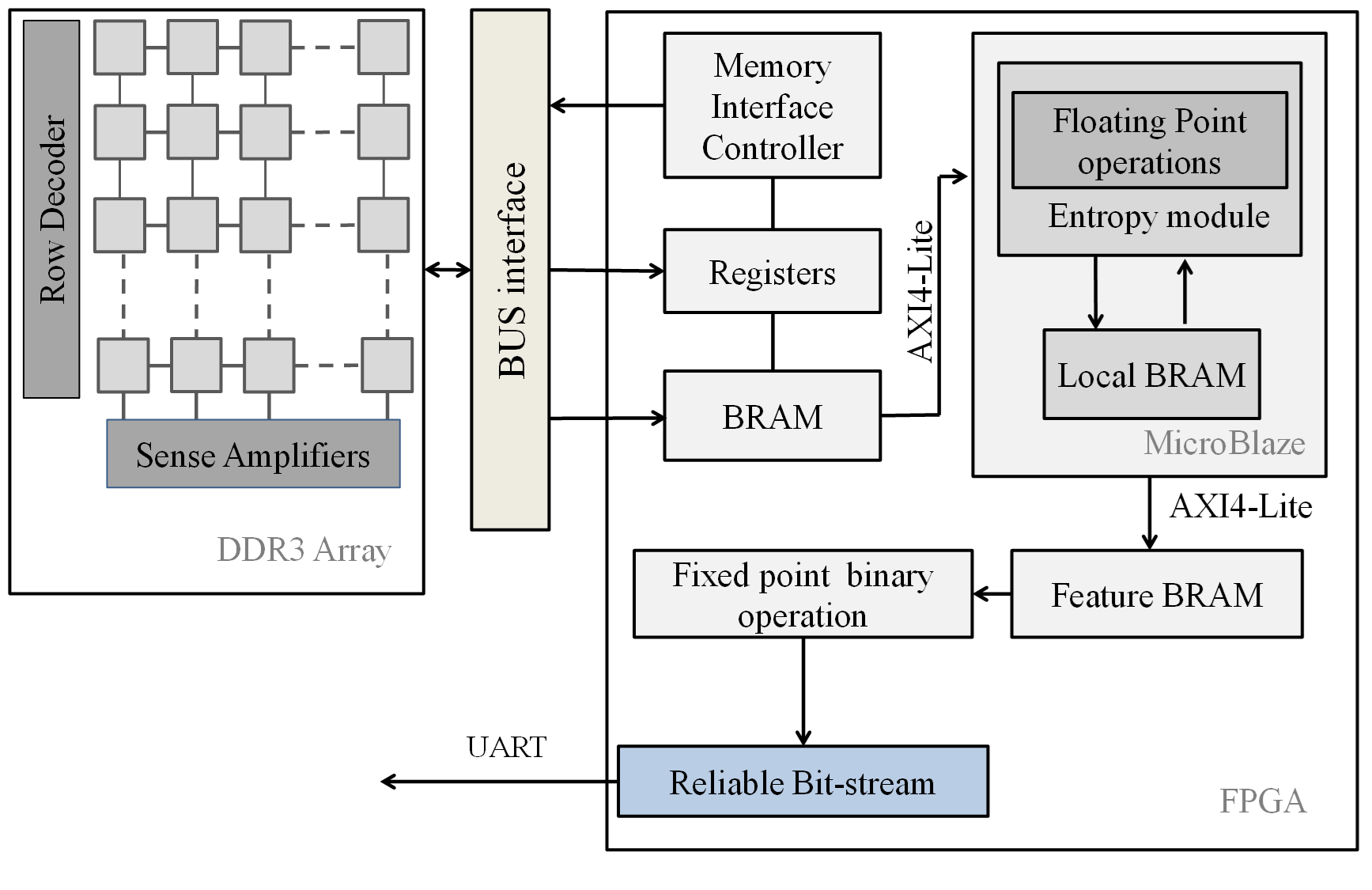}
    \caption{Architecture used for EPUF implementation.}
    \label{fig:arch}
\end{figure}
We considered a lightweight implementation and exploited a parallel design to minimize the evaluation time. To keep the logical and time overheads down, we use MicroBlaze, the Xilinx’s $32$-bit RISC core to implement the entropy module, including floating-point operations. In this architecture, the \ac{puf} segment is reserved using the specified address and the input pattern is written in the \ac{puf} segment. Then the timings are changed. Considering the size of the bitmap line, each line's data is read and is converted to integer values. Next, the probability array of integers is sent to MicroBlaze for entropy calculations via the AXI bus. After running Algorithm~1 and generating the entropy function, the results are sent to the fixed-point module to generate the final binary answer using the given precision. These steps are performed in parallel until all the lines are read. The proposed method uses low-cost and low-power FPGAs and microcontrollers. Thus, it is an efficient candidate for security of resource-constrained devices such as \ac{iot} nodes. In addition, the latency-based data set for \ac{dram} chips from different manufacturers obtained in~\cite{7} is used to verify the performance of the EPUF. We used MATLAB to simulate and evaluate our proposed system for this dataset. BRAM and UART represent Block RAM and Universal Asynchronous Receiver-transmitter, respectively.

\subsection{EPUF Evaluation}
In this section, we measure the robustness, diffuseness, and uniqueness of EPUF. We use $32$kB \ac{puf} segments to produce bitmaps of size $256 \times 125$ to generate EPUF-based responses. We exploit our proposed Algorithm~2 to extract the best precision value by reading raw data of \ac{dram} and produced EFs for $20$ times. Additionally, we execute \ref{alg:helpStream} to generate the HSs after testing EPUF responses under different temperatures along $50$ measurements. We also report the influence of different input-patterns.  
\SetKwInput{Input}{Input}
\SetKwInput{Output}{Output}
\SetKw{Break}{break}
\begin{algorithm}[!h]
    
    \labelformat{algocf}{Algorithm\,#1}

    \Input{Address and Size of the PUF segment, the Input pattern $(r)$, $dmax$ from Alg. 2, number of reads $(\Omega)$, threshold $(\Theta)$.
}
    \Output{$HS$}
    \tcc{Initialize empty array}
    $HS =  []$\;
    count $= 0$\;
    Compute $EF$ via Alg. 1 at operating temperature and apply $p$\; 
    \For{$k=1$ \KwTo $\Omega$}{
        temp = Compute $EF$ via Alg. 1 and apply $p$\;
        \For{$i=1$ \KwTo length($EF$)}{
            $X =$ xor $($temp$,EF)$\;
            \If{$X == 1$}{
                count(i) ++\;
            }
        }
    }
    \For{$i=1$ \KwTo length$(EF)$}{
        \eIf{count(i) $>\Theta$ }{
        $HS(i) = 0$\;
        }{
        $HS(i)=1$\;
        }
    }
    \Return $HS$\;
    
    \caption{Generating \ac{hs}.}
    \label{alg:helpStream}
\end{algorithm}

\textbf{Robustness.} This property shows the reproducibility of a \ac{puf} response over different measurements and operating conditions~\cite{10},\cite{11}. To examine the robustness, we input the same set of challenges to EPUF multiple times at varying temperatures ($25^{\circ}-55^{\circ}$). Then, we measured the intra-\ac{hd}s of generated responses. Fig.~\ref{fig:entr} shows the entropy values of the \ac{dram} \ac{puf} responses measured at different temperatures and their similarity to the response generated at the reference temperature ($25^{\circ}$ C). We checked the robustness considering two scenarios: 1) EFs without applying HSs, and 2) final responses after using HSs. 
\begin{figure}[!h]
    \centering
    \includegraphics[width=.9\columnwidth]{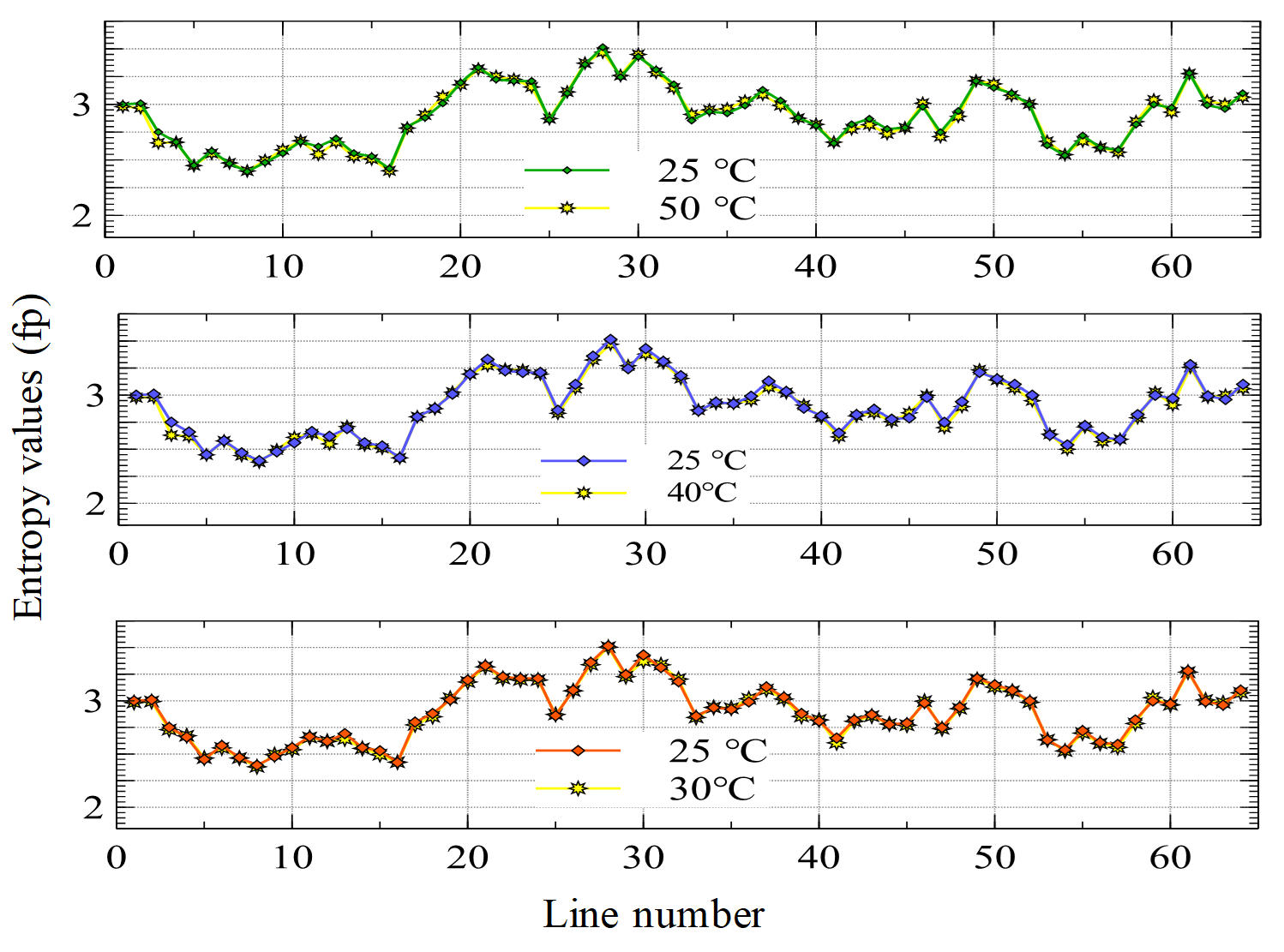}
    \caption{Entropy values of a \ac{dram} block at room temperature in comparison with other temperatures.}
    \label{fig:entr}
\end{figure}

Fig.~\ref{fig:ber_a} shows the \ac{ber} of raw \ac{dram} data. Fig.~\ref{fig:ber_b} and Fig.~\ref{fig:ber_c} show the evaluation results of EFs and final responses, respectively. From these results, we can conclude that using EFs as the responses can significantly decrease the \ac{ber} (lower than $0.013$), which greatly improves the reliability in comparison to the raw \ac{dram} data. Additionally, we eliminate the effects of unreliable features by generating HSs and applying them to the EFs, leading to reliable final responses. As shown in Fig.~\ref{fig:ber_c}, using different threshold values $\Theta$ defined in HS generation process (line 14 in \ref{alg:helpStream}) directly affects the \ac{ber}. In the case of using $\Theta=0$ as the threshold, only the entropy values will provide the final response, which include no flips over several measurements at different temperatures. This approach results in a \ac{ber} of $0$ and discards the need for any additional post-processing phases such as expensive ECC techniques. Thus, the final \ac{ber} and the number of reliable features depend on the threshold value. 
 \begin{figure}[!h]
    \subfigure[center][DRAM data]
    {\includegraphics[width = .9\columnwidth]{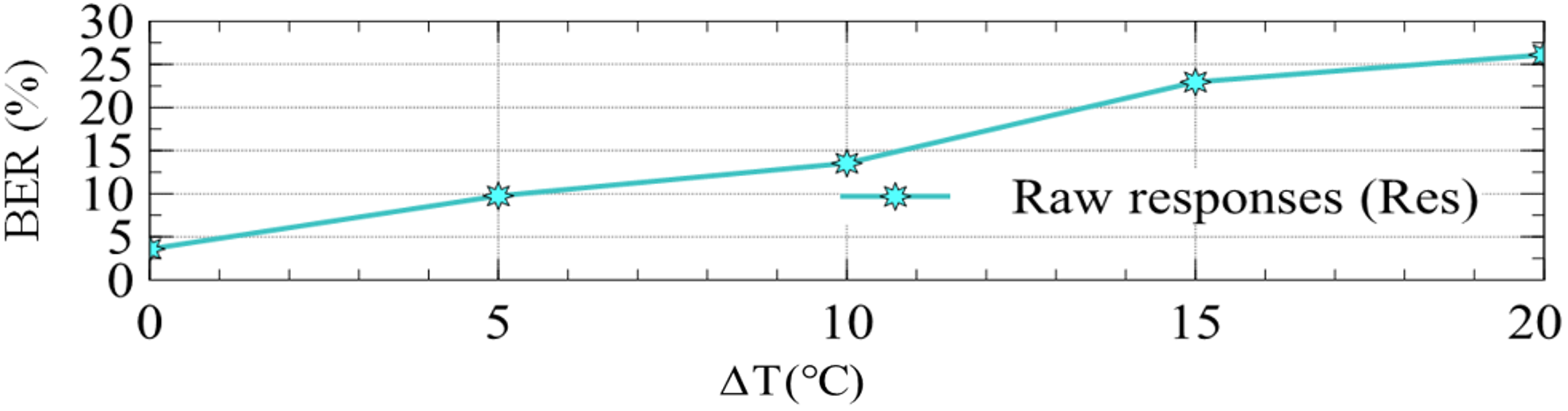}
    \label{fig:ber_a}}
    \subfigure[center][EFs]
    { \includegraphics[width = .9\columnwidth]{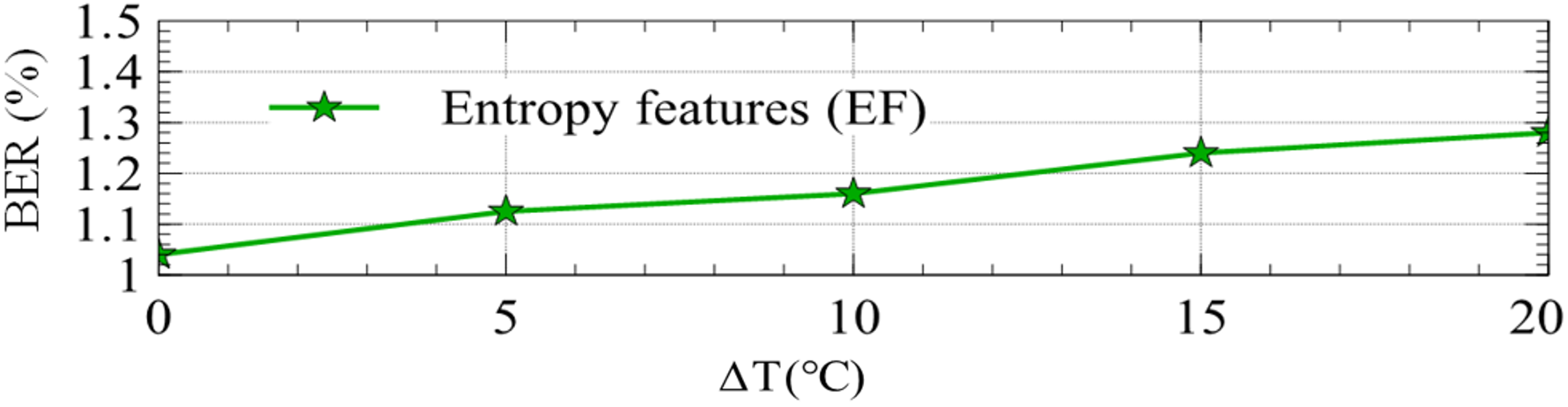}
    \label{fig:ber_b}}
    \subfigure[center][XNOR HSs with EFs]
    {\includegraphics[width = .9\columnwidth]{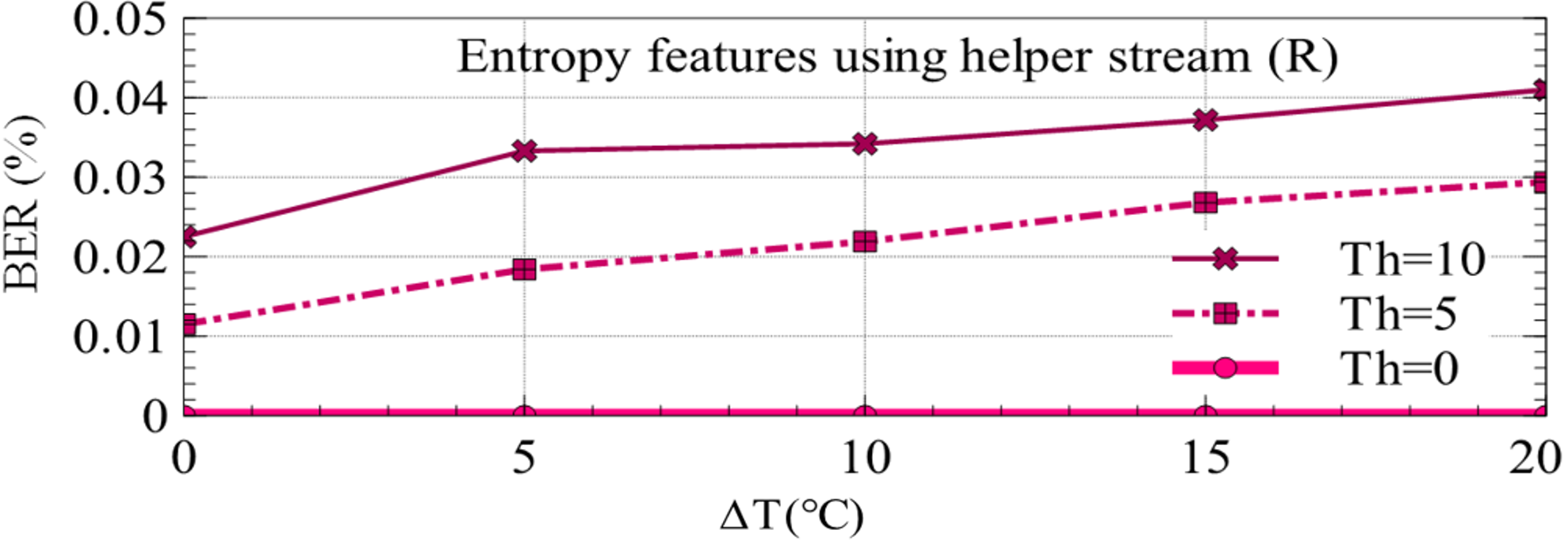}
    \label{fig:ber_c}}
 	\caption{\ac{ber} of responses: comparison of different implementations.}
 	\label{fig: ber}
\end{figure}

Fig.~\ref{fig:distr} shows the distribution of reliable bits using various threshold values. The results captured from different \ac{dram} segments show that using a lower threshold to provide lower \ac{ber} may lead to fewer qualified bits. However, the number of unqualified bits is not considerable and this issue can be dealt with using slightly larger memory segments.
\begin{figure}
    \centering
    \includegraphics[width=.9\columnwidth]{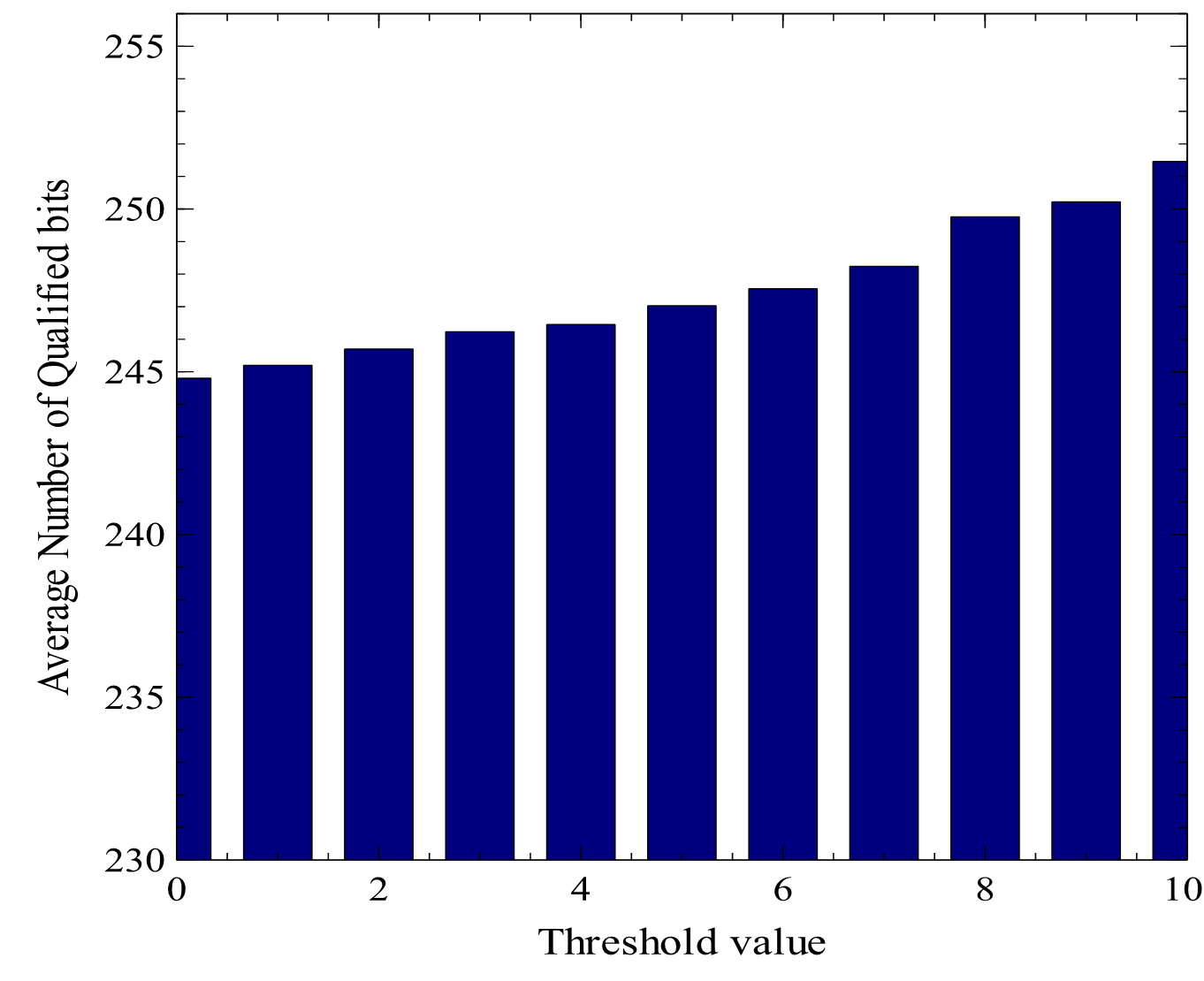}
    \caption{Distribution of reliable bits using different threshold values.}
    \label{fig:distr}
\end{figure}

\textbf{Diffuseness.} This property indicates that a \ac{puf} device can generate different and independent responses to various challenges. In this work, the input challenges are the address of \ac{puf} segments \cite{12,13,14}. We calculated the \ac{hd} between the EPUF responses corresponding to different challenges (i.e., inter-segment \ac{hd} in a same bank) to check the diffuseness. Table~\ref{tab:avg_HDs} shows the average \ac{hd} of the various responses collected from different \ac{dram} segments when the $0$ threshold is used during responses generation (see \ref{alg:helpStream}). The results indicate that the EPUF-based responses corresponding to various challenges are quite unique and independent. 
\begin{table}[!h]
\renewcommand{\arraystretch}{1.3}
\caption{Average \ac{hd}s of EPUF-based responses corresponding to different challenges}
\label{tab:avg_HDs}
\centering
\begin{tabular}{p{1.2cm}|p{2cm}||p{1.2cm}|p{2cm}}
\hline
\bfseries DRAM bank ID & \bfseries Average Inter-Segment \ac{hd} & \bfseries DRAM bank ID & \bfseries Average Inter-Segment \ac{hd}\\
\hline\hline
000 & 49.21 & 100 & 47.56 \\
\hline
001 & 39.56 & 101 & 48.84 \\
\hline
010 & 40.74 & 110 & 46.76 \\
\hline
011 & 49.25 & 111 & 46.54\\
\hline
\end{tabular}
\end{table}

\textbf{Uniqueness.} A PUF device should be identified uniquely among other devices. We quantified the uniqueness of EPUF by calculating the inter \ac{hd} of responses from different \ac{dram} chips and blocks. The distribution of inter-\ac{hd} for EFs from different \ac{dram} chips is presented in Fig.~\ref{fig:banks} and the average inter-chip \ac{hd} is $47.79\%$. Fig.~\ref{fig:chps} shows the inter-\ac{hd} results for EPUF responses provided by different banks of the same \ac{dram} chip when the threshold $\Theta$ in \ref{alg:helpStream} is set to $0$. The average inter-bank \ac{hd} of EPUF responses is $47.09\%$.  Therefore, EPUF can provide desirable uniqueness as well as improved reliability.
 \begin{figure}[!h]
    \subfigure[center][Different banks]
    { \includegraphics[width = 0.45\columnwidth]{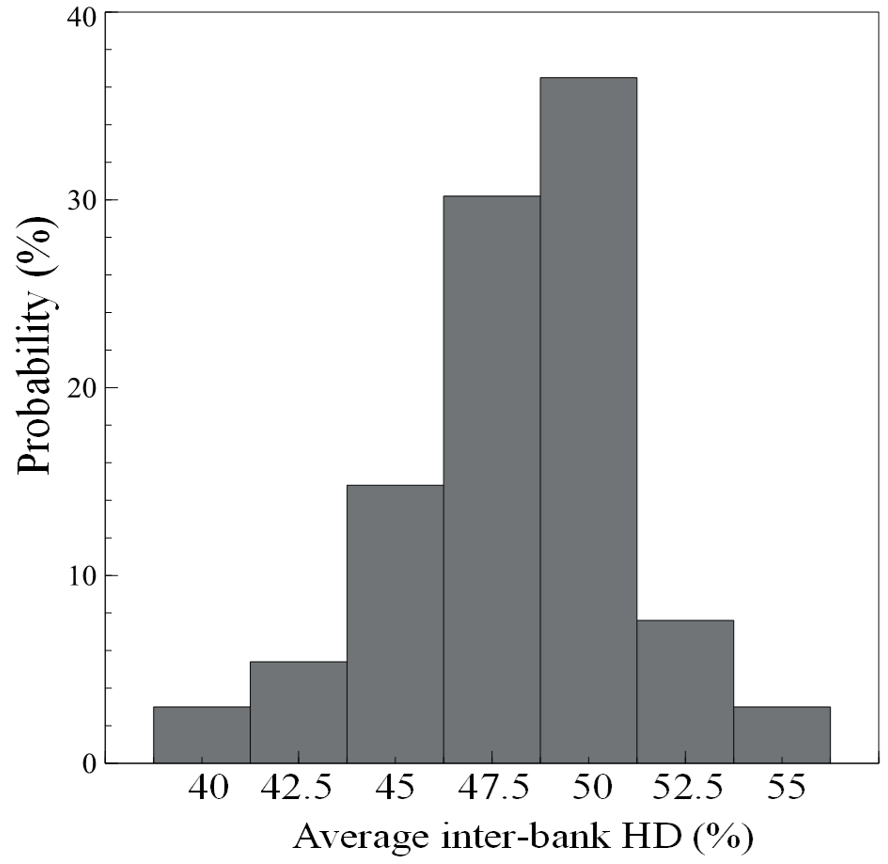}
    \label{fig:banks}}
    \subfigure[center][Different chips]
    {\includegraphics[width = 0.45\columnwidth]{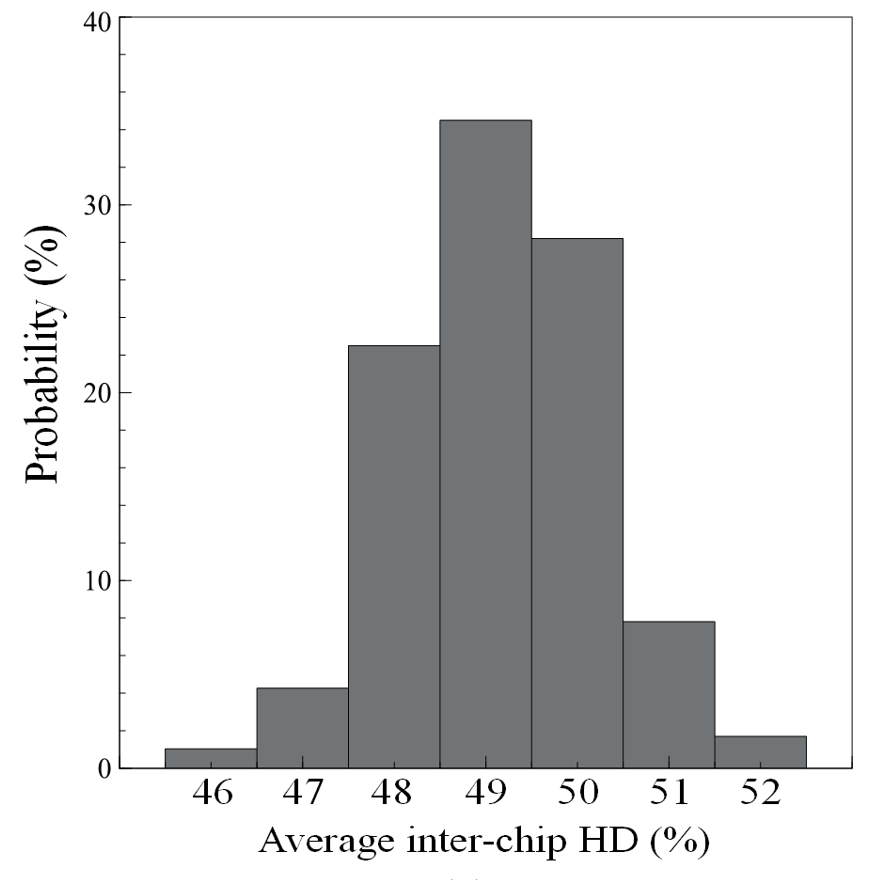}
    \label{fig:chps}}
 	\caption{Distribution of inter-\ac{hd}.}
 	\label{fig: interHD}
\end{figure}

\textbf{EPUF Performance Evaluation.} We implemented the EPUF architecture such that \ac{dram} read operations and responses generation can be run in parallel. The worst evaluation period, taking into account all the steps needed to produce at least $700$-bit accurate response, is experimentally calculated to be $0.78$ms. The comparison of evaluation time of EPUF with the existing latency-based \ac{dram} PUFs is shown in Table~\ref{tab:evTime}. We notice that the results obtained using EPUF are significantly lower than those obtained using other methods. Hence, EPUF minimizes system interference and can act as a fast and real-time security primitive for various applications.
One of the important challenges in existing \ac{dram} PUF mechanisms is the limitations of qualified segments to be used as PUF segments. This is due to the fact that the number of bit failures in the reduced timing values is not distributed uniformly within the whole \ac{dram}. In this way, only segments with efficient bit failures are known as the PUF segments and the \ac{crp} space is significantly limited. This issue can be resolved by using entropy features for response generation instead of directly exploiting \ac{dram} content. Furthermore, due to the parallelism used to read and process \ac{dram} rows, we interface only one rank during the PUF operation. Therefore, this method reduces the hardware and memory space overheads. 
Our proposed mechanism uses low-cost and low-power FPGAs or microcontrollers with floating-point units, which are embedded with the memory controller enabling the possibility of timing parameter access during system run-time.
\begin{table}[!h]
\renewcommand{\arraystretch}{1.3}
\caption{Comparison of EPUF evaluation time with other existing technologies}
\label{tab:evTime}
\centering
\begin{tabular}{p{4cm}||p{3cm}}
\hline
\bfseries \ac{dram} PUF Technology & \bfseries Evaluation Time [ms] \\
\hline\hline
Reduced $t_{RCD}$-based PUF \cite{6} & $88.2$ \\
\hline
Fast \ac{dram} PUF \cite{7} & $3.19$   \\
\hline
PreLat PUF \cite{8} & $1.59$   \\
\hline
Proposed EPUF & $0.93$ \\
\hline
\end{tabular}
\end{table}

\section{EPUF-Based Key Generation}
PUFs exploit the physical circuit properties to produce responses and, like any physical measurement, are eventually influenced by varying operating conditions. Thus, PUF response reproductions are not fully stable. Therefore, most PUF-based key generators use post-processing mechanisms such as fuzzy extractors to transform unstable responses into a secure, reliable, and uniform cryptographic key~\cite{15, 16, 17, 18}. This method requires extra hardware area to implement the helper data algorithms and ECC logic, and also leads to additional computational costs. Therefore, these types of PUF key generators are not appropriate for use on low-power computing platforms such as resource constrained \ac{iot} devices. On the other hand, \ac{dram} PUFs use pre-processing mechanisms to allocate reliable \ac{dram} cells and filter noisy cells to generate PUF responses. To the best of our knowledge, there is no \ac{dram} PUF-based key generator implementation that provides reliable keys for cryptographic protocols without the need for ECC techniques. 
EPUF-based responses can be used to extract secure and random keys with high entropy. Additionally, the high reliability of these responses make them suitable for key agreement applications. In particular, we consider the use of EPUF responses as secure \ac{dram} primitives to derive a secure and reliable cryptographic key, which is operational on low cost and low power microcontrollers.

Fig.~\ref{fig:keyGen} illustrates the EPUF-based key generation, including registration and reconstruction phases. During the registration, the server records multiple responses that are generated from the same challenge under different operating conditions. These responses are applied to our characterization algorithms to generate EF, HS, and finally the reliable response $R$. Then, the secure key is generated using one-way hash function $SK= h(R)$, which is exploited to produce a full bit random cryptographic key and allows cryptographic protocols to be operational and resilient against various attacks. 
\begin{figure}[!h]
    \centering
    \includegraphics[width=\columnwidth]{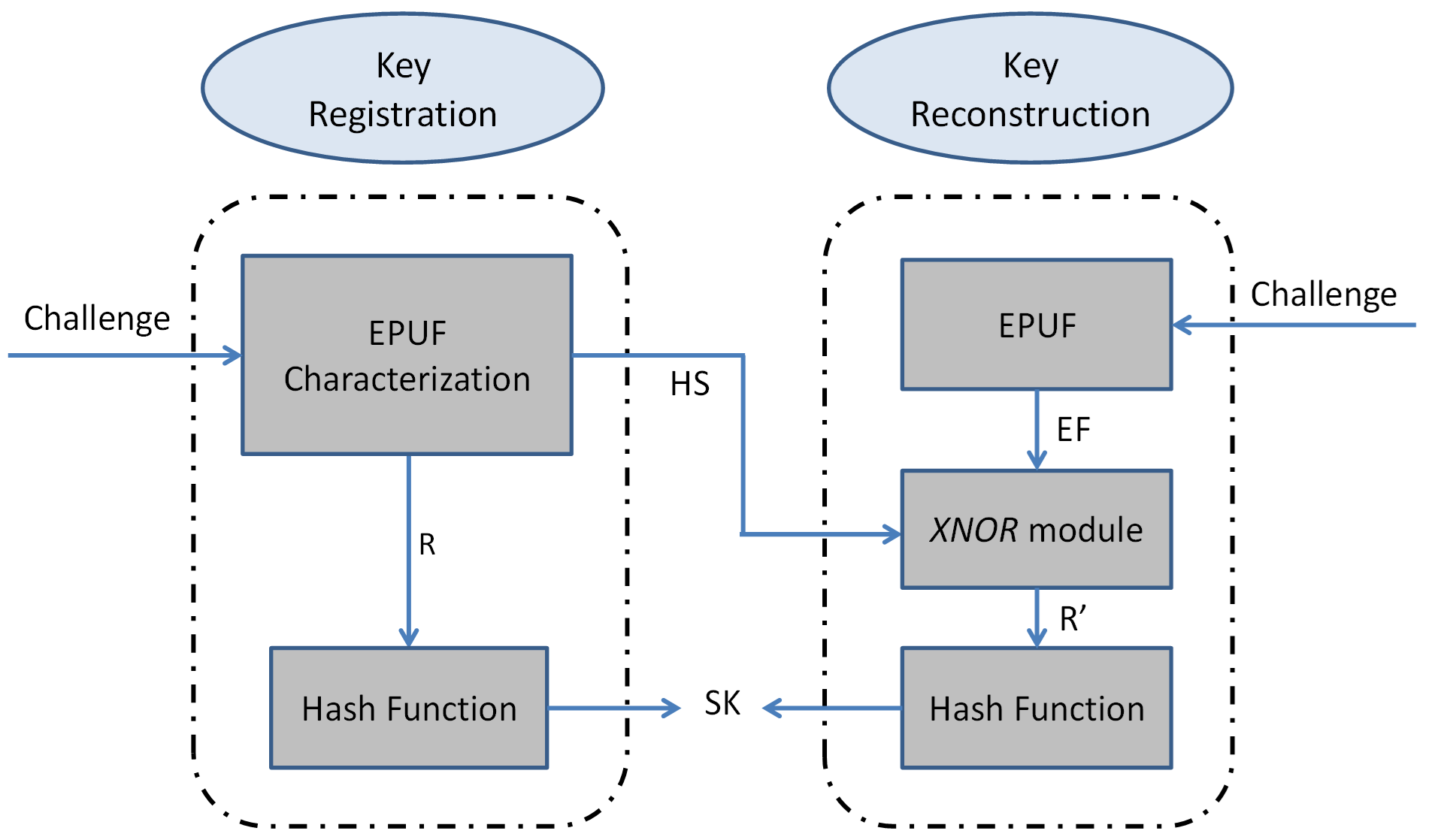}
    \caption{EPUF-based key generation. }
    \label{fig:keyGen}
\end{figure}
To perform key reconstruction, the raw response (Res) is read from \ac{dram} when the PUF is deployed in the field. Next, the read response is given to entropy module to generate the corresponding EF. Then, the final response is determined by XNORing EF and HS $(R=EF \odot HS),$ and the SK is generated after applying the one-way hash function.  

\section{EPUF-Based Authentication for \ac{iot} Networks}
In this section, we present a lightweight EPUF-based authentication protocol for communication between an \ac{iot} device and a server in a typical \ac{iot} network. EPUF allows for robust authentication and identification protocols without needing ECC.

\subsection{System and Threat Model }
We consider an \ac{iot} system consisting of many \ac{iot} devices and one server. An \ac{iot} device is a node in the network that is equipped with the proposed EPUF and securely interacts with its own environment. This node can send reports to the server and receive control messages accordingly. Although each \ac{iot} node is a highly resource-constrained device located in the open environment without considerable hardware protection, EPUF is embedded in the \ac{iot} device and any physical attempt to remove the EPUF from the device leads to the destruction of both device and EPUF. The server controls all \ac{iot} devices in the network and has a large database, high computational power, and sufficient hardware protection.

In this paper, we consider the widely-used Canetti-Krawczyk (CK)~\cite{19} adversary model. In CK (as a stricter model than Dolev–Yao model~\cite{20}), the adversary has full access to the communication link. Thus, the adversary can perform eavesdropping, modifying, impersonation, and replay attacks, or disrupt the network by performing Denial of Service (DoS) attacks. Furthermore, the adversary can obtain the secrets stored in the memory of the deployed \ac{iot} devices. In addition, we consider the outcome of ephemeral secrets leakage of one session proposed in the CK-adversary model. We propose an EPUF-based authentication protocol for \ac{iot} environment that not only resists all the possible attacks, but also provides better security and efficiency features compared to the state-of-the-art.

\subsection{Registration Phase }
Similar to the registration phase in EPUF-based key generation, the $j$-th \ac{iot} device, $Dev_j$, determines the initial parameters by running the EPUF characterization algorithms for multiple challenges. Since we are using most of the banks in the embedded \ac{dram} for the proposed EPUF, a large set of \acp{crp} is generated per device in this phase. Then, $Dev_j$ sends a request to the server and a row in the server's database is generated consisting of the exploited challenges, HSs, and the hash of the corresponding golden responses. We assume that the golden response $R$ is the final reliable response, which is produced via the XNOR of EF and HS. Furthermore, $Dev_j$ randomly generates an identifier, $ID_j^0$, and sends all the parameters to the server. Then $Dev_j$ removes all parameters including the \acp{crp} from its memory, except for $ID_j^0$. We assume that $n$ is the maximum number of generated \acp{crp}, and $\ell$ is the effective number for $0\leq \ell \leq n$. The server stores the $Dev_j$'s parameters including all the challenges ($Ch_{\ell}^j$), all the related keys ($K_{\ell}^j$), the \ac{hs}s ($HS_{\ell}^j$), and the pseudo-ID $ID_j^0$ in its database. This step can be repeated over long periods depending on the policies of the system and the number of generated \acp{crp}. Furthermore, all messages are communicated through a secure channel during this phase. 

\subsection{Authentication Phase }
In this section, we only consider the $i$-th communication between $Dev_j$ and the server, which can be easily extended to all devices at different times. The main basis of authentication in this paper is the mutual validation of the protocol parties based on the comparison of PUF responses to the same challenge input. Fig.~\ref{fig:auth} shows our proposed authentication scheme based on EPUF. We provide more details of the proposed authentication protocol in the following sections.
\begin{figure}[!h]
    \centering
    \includegraphics[width=\columnwidth]{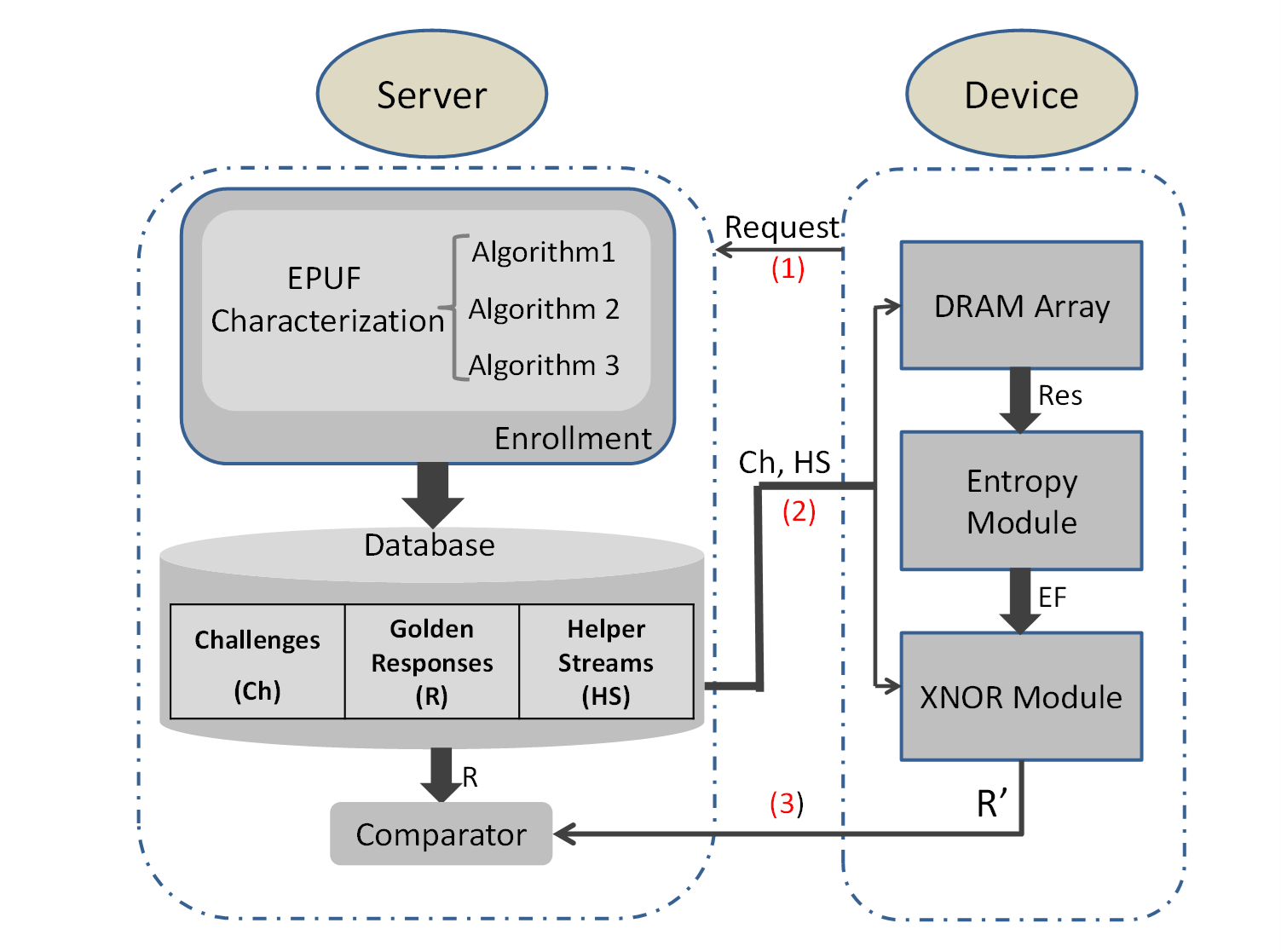}
    \caption{EPUF-based authentication.}
    \label{fig:auth}
\end{figure}

To start the authentication process, $Dev_j$ generates a nonce $N_j^i$  and sends it with $ID_j^i$ to the server. After receiving this message, the server looks for $ID_j^i$ in its database and finds the corresponding challenge, key, and HS, ${Ch_i^j, K_i^j, HS_i^j}$. Next, it generates a random number $rnd_i^j$ and computes $A_i^j=((M_i^j \oplus rnd_i^j ) || rnd_i^j) \oplus K_i^j$, where $M_i^j$ is the control message from server to $Dev_j$, $\oplus$ represents the XOR operation, and $||$ denotes the concatenation operation. The server also computes a hashed value as its verifier $V1_i^j=h(M_i^j, rnd_i^j, N_i^j, HS_i^j,  Ch_i^j, K_i^j )$ and then sends the messages ${A_i^j, V1_i^j, Ch_i^j, HS_i^j}$ to $Dev_j$.   
Upon receiving the message, $Dev_j$ runs its EPUF as $ R_i^j=EPUF(Ch_i^j )\odot HS_i^j$ to obtain the golden response and then computes the key as $K_i^j=h(R_i^j )$. Next, $Dev_j$ decrypts $A_i^j$ to obtain $M_i^j$ and $rnd_i^j$ with only one XOR operation $A_i^j \oplus K_i^j$. $Dev_j$ then checks if $V1_i^j=h(M_i^j, rnd_i^j, N_i^j, HS_i^j, Ch_i^j, K_i^j )$ holds. If $V1_i^j$ passes the verification, the server is authenticated by $Dev_j$, otherwise the message will be discarded. Then, $Dev_j$ generates a pseudo-random number $ID_j^{i+1}$, as its new identifier, and computes $E_i^j=(D_i^j||ID_{i+1}^j ) \oplus h(K_i^j )$, where $D_i^j$ is the $i$-th message of $Dev_j$ to the server. Finally, $Dev_j$ computes $V2_i^j=h(D_i^j, rnd_i^j, ID_{i+1}^j, K_i^j )$ for the verification at the server side, removes all the parameters but $ID_j^{i+1}$ from its memory, and sends $\{E_i^j, V2_i^j \}$ to the server. The server can decrypt $D_j^i$ and $ID_{i+1}^j$ by computing $E_j^i\oplus h(K_j^i )$ from the receiving packet, and then checks if $V2_i^j=h(D_i^j, rnd_i^j, ID_{i+1}^j, K_i^j )$ holds. If $V2_j^i$ passes the verification process, the mutual authentication is established. The server accepts the message and replaces $ID_{i+1}^j$ with the previous one in its database. Otherwise, the message will be discarded. Our proposed EPUF-based authentication protocol is depicted in Fig.~\ref{fig:authProt}.
\begin{figure}
    \centering
    \includegraphics[width=0.7\columnwidth]{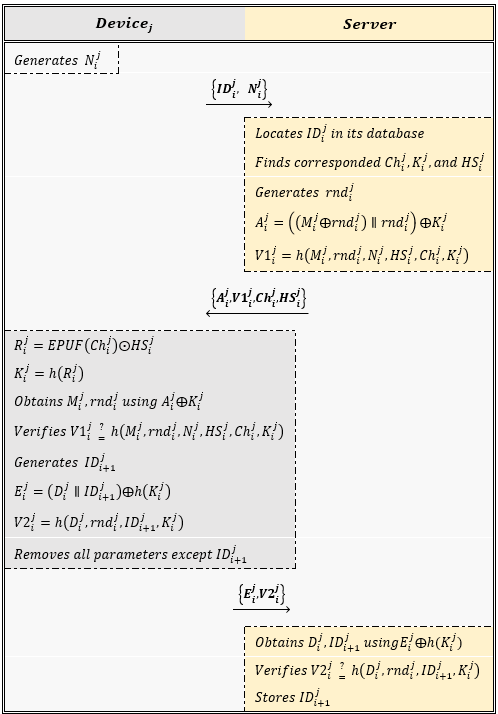}
    \caption{The proposed EPUF-based authentication protocol.}
    \label{fig:authProt}
\end{figure}

\subsection{Formal Security Analysis }
In this section, we prove the security of our protocol in the presence of Probabilistic Polynomial Time (PPT) adversaries. Notice that the authenticity of the scheme only relies on using a secure and collision-resistant one-way hash function. Therefore, we only informally discuss the authenticity of the proposed protocol in the next section. In the following, we prove the confidentiality of our protocol~\cite{21}. We first introduce two lemmas:
\begin{lemma}
$EPUF(\cdot)$ function is a true random number generator and the value of $R_i=EPUF(C_i )$, where $C_i$ is an arbitrary challenge, is indistinguishable from a same length string created by a pseudo random generator (PRNG) for a PPT adversary. 
\end{lemma}
\begin{proof}
The proof comes from the randomness property of the proposed EPUF. 
\end{proof}
\begin{lemma}
The one-way collision-resistant hash function $h(\cdot)$ with random input and a uniform distribution is indistinguishable from a same-length string created by a PRNG for a PPT adversary.  
\end{lemma}
\begin{proof}
In this game-based proof, two random numbers $p$ and $q$ are generated. Next, a coin is flipped and based on which side it lands, the parameter $z$, which can be either $h(p)$ or $q$, will be given to a distinguisher adversary $\mathbb{D}$. $\mathbb{D}$’s purpose is to find out whether $z$ is a hashed value created by $h(\cdot)$ or is a random string generated by a PRNG. Thus, if the probability of $\mathbb{D}$’s success is negligibly more than a pure guess $(1/2+\epsilon(n))$. Thus, we have to prove that the following equation holds: 
\begin{equation}
    |pr[\mathbb{D}^{z=h(p) }=1]-pr[\mathbb{D}^{z=q}=1]| \leq negl (n),
\end{equation}
where $negl(n)$ is a negligible value.

By definition, when the random string is generated by a PRNG $g$; $p \gets g(seed)$, the probability of success for $\mathbb{D}$ is equal to a pure guess i.e., $pr[\mathbb{D}^{z=q}=1]=1/2$ . If $z$ is a hashed value of $h(p)$, where $p$ is randomly and uniformly selected, the probability of $\mathbb{D}$’s success is $pr[\mathbb{D}^{z=h(p) }=1]=1/2+\epsilon(n)$. Notice that $1/2$ comes from the probability of a pure guess and $\epsilon(n)$ shows all the additional information that $\mathbb{D}$ can obtain from $h(p)$ with a randomly-selected $p$. Thus, we have:
\begin{equation}
\begin{split}
    \epsilon(n) & = pr[find p' | h(p' )=z]= pr[find p' |h(p' )=h(p)]= \\ 
    & pr[find p' |h(p' )=h(p),p'\neq p] pr[p'\neq p]+ \\
    & pr[find p' |h(p' )=h(p),p'=p] pr[p'=p].
\end{split}
\end{equation}
As $p$ and $p'$ are randomly generated, $pr[p'=p]=1/2^{l(n)} =negl (n)$, where $l(n)$ is the length of the random string. Furthermore, the probability $pr[find p' |h(p' )=h(p),p'\neq p]$ is equal to find a collision in $h(\cdot)$, which is considered to be negligible. Therefore:
\begin{equation}
 \begin{split}
     & \epsilon(n) = 
      negl(n)×(1-1/2^{l(n)} )+ \\
      & pr[find p' |h(p' )=h(p),p'=p] \cdot negl (n) \leq negl (n).
 \end{split}   
\end{equation}
So we have $pr[D^{z=h(p) }=1]\leq 1/2+negl (n)$. As a result, the following equation is obtained as a proof for lemma 2. 
\begin{equation}
\begin{split}
    & |pr[D^{z=h(p) }=1]-pr[D^{z=q}=1]|=\\
    & |1/2+negl (n)-1/2| \leq negl (n).
\end{split}    
\end{equation}
\end{proof}
Thus, if a PPT distinguisher $\mathbb{D}$ is able to distinguish between the output of a one-way collision-resistant hash function and a random string, then another adversary can find a collision in the mentioned one-way hash function.
\begin{theorem}
The proposed protocol has indistinguishable encryption security in the presence of PPT adversaries. 
\end{theorem}
\begin{proof}
We present the proof for $A_i^j$ in the following, and the same proof can be used for $E_i^j$. First, by naming the proposed protocol $\pi$ we have:
\begin{equation}
    \pi = 
    \begin{cases}
        Gen & k \gets \{0,1\}^n, \\ 
        Enc & C = m \oplus G(k), \\
        Dec & m = C \oplus G(k);
    \end{cases}
\end{equation}
where $k$ is considered as a challenge, $G(k)=h(EPUF(k))$ is the cryptographic key, $C=E_i^j$ is the encrypted message, and $m=(M_i^j \oplus rnd_i^j )|| rnd_i^j$ is considered as the plaintext. For the proof, we introduce a new game as shown in Fig.~\ref{fig:adv}. The game is based on two PPT adversaries: an eavesdropper $A$ who attacks the proposed protocol, and a distinguisher $\mathbb{D}$ who uses the information of $A$ to distinguish between two random strings.
\begin{figure}
    \centering
    \includegraphics[width=0.7\columnwidth]{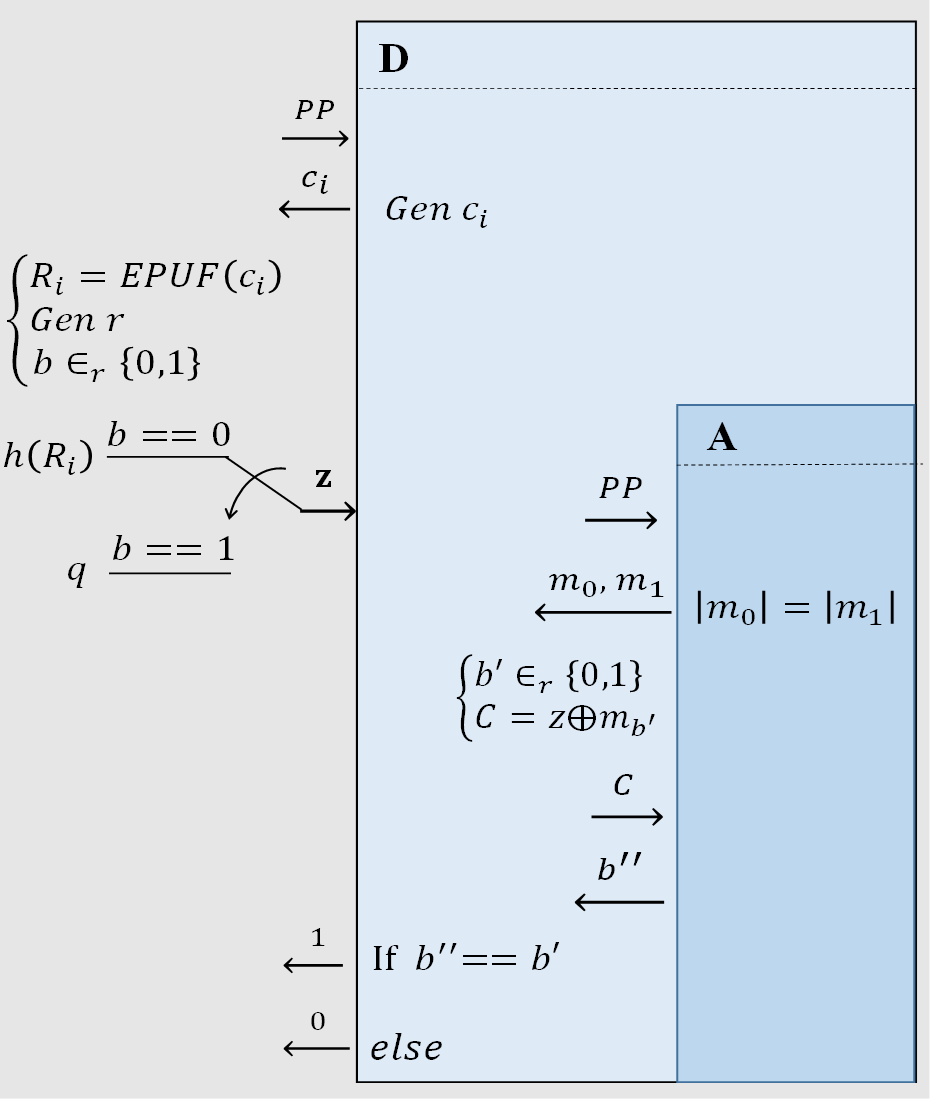}
    \caption{The adversary model in the provided game.}
    \label{fig:adv}
\end{figure}
After receiving the public parameters $\{l(n),h(\cdot)\}$ by both $A$ and $\mathbb{D}$, $A$ selects two arbitrary messages $m_0,m_1$ with the same length and sends them to $\mathbb{D}$. Then, $\mathbb{D}$ sends a challenge $c_i$ to the Oracle. The Oracle flips a coin and based on the outcome, sends $h(EPUF(c_i ))$ or a random string $q$ generated by a PRNG. $\mathbb{D}$ also flips a coin and encrypts either $m_0$ or $m_1$ for sending one of them to $A$. $A$ succeeds if it is able to find $b'$ with probability higher than a pure guess $(1/2)$. 

We prove that if an adversary like $A$ is able to distinguish $C$ and find $b'$ with non-negligible probability $\epsilon$, then an adversary like $\mathbb{D}$ is able to distinguish between a hashed value with a random input and a random string created by PRNG. However, this is a violation of lemma 2. First, consider the case  $b==1$ and hence, $z=q$. $\mathbb{D}$ choses $b'\in_r {0,1}$, encrypts the message $m_{b' }$ by computing $C=q\oplus m_{b'}$, and sends it to $A$. $A$'s purpose is to determine $b'$ to find out which of the $m_i$ is encrypted. Thus, the probability of successfully guessing which of the messages is encrypted is equal to pure guess for the PPT adversary $A$, i.e.,
\begin{equation}
    pr[Privk_{A,\pi}^{eav} (n,q)=1]=\frac{1}{2}.
\end{equation}

Assume that $b==0$ thus, $z=h(R_i )$ and the encrypted message is computed as $C=h(R_i)\oplus m_{b' }$. We will see what happens if $A$ can succeed and finds $b'$ with probability $1/2+\epsilon(n)$, where $\epsilon(n)$ is not negligible. Hence, we consider:
\begin{equation}
    pr[Privk_{A,\pi}^{eav} (n,h(R_i ))=1]=1/2+\epsilon(n).
\end{equation}
Notice that $A$ directly attacks the protocol and the probability of success means the security failure of our scheme. We calculate the probability of $\mathbb{D}$’s success which is the ability of distinguishing between $h(R_i)$ and $q$. Thus:
\begin{equation}
    \begin{split}
        & |pr[D^{h(R_i )} =1]-pr[D^q=1]|= \\
        & |pr[Privk_{A,\pi}^{eav} (n,h(R_i ))=1]-pr[Privk_{A,\pi}^{eav} (n,q)=1]| \\
        &\leq |1/2+\epsilon(n)-1/2|=\epsilon(n),
    \end{split}
\end{equation}
which violates Lemma 2. As a result, we have proved that if an adversary/eavesdropper like $A$ is able to attack the proposed protocol, then an adversary/distinguisher like $\mathbb{D}$ can distinguish between a random string and a hashed value with an unknown and random input. Therefore, according to Lemma 2, if such $\mathbb{D}$ exists then, another adversary exists who is able to find a collision in the secure one-way hash function. The indistinguishable encryption security of the proposed protocol has been proved on the device’s side. The server’s side security can be brought into the same analysis as well.
\end{proof}

\subsection{Informal Security Analysis }
In this section, we discuss the security of our protocol against the adversarial attacks in the CK-adversary model.

\textbf{Message Analysis Attack.} The adversary tries to undermine the confidentiality of messages $A_i^j$ and $E_i^j$ to obtain the pseudo-random numbers or the plaintexts. However, $A_i^j$ and $E_i^j$ are respectively encrypted with the secure EPUF-based keys i.e., $K_i^j$ and $h(K_i^j )$, which are one-time pad keys and securely derived from the server's database. Thus, the adversary is not able to decrypt any secret from the transmitted messages by eavesdropping the communication link. As a result, this kind of attack is computationally infeasible for a PPT adversary and our scheme provides a good level of confidentiality. 

\textbf{Message Altering and Impersonation Attack.} The adversary tries to modify the communicated messages or impersonate as one of the legal parties of the protocol. Since our protocol supports two-way communication, an adversary may impersonate both \ac{iot} devices and server. In both cases, the adversary aims to modify or inject the messages and then create his/her responding verifier in a way that it passes the one-way hash-based verification process at the other end. However, because of the one-way and collision-resistant properties of the cryptographic hash functions used in constructing $V1_i^j=h(M_i^j,rnd_i^j,N_i^j,HS_i^j,Ch_i^j,K_i^j )$ and $V2_i^j=h(D_i^j,rnd_i^j, ID_{i+1}^j,K_i^j )$ and also lack of sufficient information about $rnd_i^j$ and $K_i^j$, performing these kind of attacks would be computationally infeasible for a PPT adversary.

\textbf{Replay Attack.} The adversary tries to use the out-of-date packets, which were previously used in former authentication phases, and replay them to each of the parties in the network. However, both \ac{iot} device and server generate a new nonce and random number for each authentication phase, and use them in their verifies as $V1_i^j=h(M_i^j, rnd_i^j, N_i^j, HS_i^j, Ch_i^j, K_i^j)$ and $V2_i^j=h(D_i^j, rnd_i^j, ID_{i+1}^j, K_i^j)$, respectively. Hence, the verification processes will fail when the adversary replays the out-of-date packets. In addition, the session keys are securely created based on EPUF which are updating per each protocol runtime. 

\textbf{DoS Attack.} The adversary's purpose is to exhaust all network resources such that legitimate parties are unable to process useful information. Since the server has high computational power, we consider DoS attack only on device side. This attack is done on the device side when the server sends the packet ${A_i^j, V1_i^j, Ch_i^j, HS_i^j }$ to $Dev_j$. For that, the adversary sends bogus messages such as ${W, X, Y, Z}$ in form of ${A_i^j, V1_i^j, Ch_i^j, HS_i^j }$ to $Dev_j$. Consequently, all $Dev_j$ needs to do for discarding this message is one EPUF evaluation and XNOR operation to retrieve $R_i^j$, one one-way hash function to generate $K_i^j$, one XOR operation to decrypt $A_i^j$, and another one-way hash function for verifying $V1_i^j$. Therefore, this device can reject the received message without consuming considerable computational resources. Note that, the probability that the random tuple ${W,X,Y,Z}$ successfully passes the verification process is negligible.

\textbf{Privacy and Anonimity.} $Dev_j$ generates a new ID $(ID_{i+1}^j)$ for each (next) authentication phase and securely sends it to the server. Furthermore, $Dev_j$ uses its pseudo-ID only once and no entity except $Dev_j$ and the server knows about $ID_{i+1}^j$; only the server knows about the devices' activities. As a result, the proposed scheme provides privacy and anonymity for the \ac{iot} devices, not only against the external adversary but also against the other trusted entities in the network. Furthermore, even the server does not know about the devices' IDs, unless the devices send their IDs to the server during the authentication process.

\textbf{Perfect Forward/Backward Secrecy.} We investigate whether revealing of one session key $(K_i^j)$ may lead to the inference of previous keys. According to Fig.~\ref{fig:authProt}, the $i$-th and $i+1$-th session keys are computed as $K_i^j=h(R_i^j)$ and $K_{i+1}^j=h(R_{i+1}^j)$, respectively. $R_i^j$ and $R_{i+1}^j$ are two different responses derived from our proposed EPUF, thus, the session keys $K_i^j$ and $K_{i+1}^j$ are completely independent. Thus, the adversary will not obtain any knowledge about the other previous session keys by having one. The same analysis implies when the adversary tries to obtain the next keys by having one session key. In fact, since all the keys are stored in the server's database through a secure channel in the registration phase, the adversary has no chance to obtain the previous/next session keys using the compromised key. Therefore, the proposed protocol provides perfect forward/ backward secrecy. 
\begin{table*}[!h]
\renewcommand{\arraystretch}{1.3}
\caption{Security and performance comparisons between the proposed scheme and the state-of-the-art (\cmark= has the feature, \xmark= does not have the feature, n.a. $=$ not applicable).}
\label{tab:eval}
\centering
\begin{tabular}{|p{1cm}|p{2cm}|p{3cm}|p{1.5cm}|p{0.4cm}|p{0.4cm}|p{0.4cm}|p{0.4cm}|p{0.4cm}|p{0.4cm}|p{0.4cm}|p{0.4cm}|p{0.4cm}|p{0.4cm}|}
\hline
\bfseries Scheme & \bfseries Communication Overhead [B] & \bfseries Used Primitives & \bfseries Computation Cost [ms] & \bfseries $F_1$ & \bfseries $F_2$ & \bfseries $F_3$ & \bfseries $F_4$ & \bfseries $F_5$ & \bfseries $F_6$ & \bfseries $F_7$ & \bfseries $F_8$ & \bfseries $F_9$ & \bfseries$F_{10}$\\
\hline\hline
\cite{24} & $336$ & $T_{\rm PUF} + T_{\rm REC} + 14 T_{\rm H} + 2T_{\rm RN} + 2 T_{\rm PM} + 2T_{\rm PA} + T_{\rm ME}$ & $25.314$ & n.a. & \cmark & \xmark & \cmark & \xmark  & \cmark & \xmark & \cmark & \xmark  & \xmark \\
\hline
\cite{25} & $260$ & $T_{\rm PUF} + T_{\rm REC} + 6 T_{\rm H} + 7 T_{\rm PM} + 2T_{\rm PA}$ & $46.656$ & n.a. & \cmark & \xmark & \cmark & \xmark  & \cmark & \xmark & \cmark & \xmark  & \xmark \\
\hline
\cite{26} & $176$ & $2T_{\rm PUF} + 3 T_{\rm HMAC} + T_{\rm RN} + 4 T_{\rm SE} + 5T_{\rm SD}$ & $2.881$ & \cmark & \cmark & \cmark & \cmark & \xmark  & \xmark & \xmark  & \xmark & n.a. & \cmark \\
\hline
\cite{27} & $240$ & $2T_{\rm PUF} + T_{\rm GEN} + 7 T_{\rm HN} + T_{\rm RN}$ & $3.501$ & n.a. & \cmark & \cmark & \cmark & \cmark  & \xmark & \xmark  & \cmark & \xmark & \cmark \\
\hline
\cite{28} & $224$ & $2T_{\rm PUF} + 3 T_{\rm HMAC} + 2 T_{\rm H} + 2 T_{\rm RN} + 3 T_{\rm SE} + 2 T_{\rm SD}$ & $2.112$  & \cmark & \cmark & \cmark & \cmark  & \xmark & \xmark  & \xmark & \xmark & n.a. & \cmark \\
\hline
Ours & $192$ & $T_{\rm EPUF} + 4 T_{\rm H} + 2 T_{\rm RN} $ & $1.134$  & \cmark & \cmark & \cmark & \cmark  & \cmark & \cmark  & \cmark & \cmark & \cmark & \cmark \\
\hline
\end{tabular}
\end{table*}

\textbf{Ephemeral Secret Leakage Attack.} The short term secrets of an arbitrary session may be revealed in the CK-adversary model. Hence, the adversary’s purpose is to obtain as much information as possible to attack the protocol in other sessions. In our protocol, the short term secret $rnd_i^j$ is randomly generated by the server at each session. If the adversary knows $rnd_i^j$, he/she can only obtain the right half of the session key $K_i^j$ from $A_i^j$. However, as the creation of the next key is completely independent from the previous keys, the other session keys will not be revealed to the adversary. Another short term secret in this protocol is $Dev_j$'s ID, $ID_{i+1}^j$, whose disclosure at every session leads to the divulgation of the right half of $h(K_i^j)$, which will not reveal any additional information. Since our protocol supports perfect forward/backward secrecy, the disclosure of the session keys will not reveal the other ones. 

\textbf{Physical Attack.} We assume that only the device's \ac{nvm} is vulnerable to an adversary. Hence, when the adversary physically captures an \ac{iot} device, he/she is able to read the secrets (e.g. the session key from the device's \ac{nvm}). Notice that the volatile memory used to store the intermediary secret values during the authentication phase is wiped out after the end of each communication and is not accessible by an adversary~\cite{22,23}. As seen in our protocol, the session keys are generated via EPUF in each authentication process and there is no need for devices to store the secret keys. Therefore, when an \ac{iot} device is physically captured by an adversary, he/she doesn't obtain any secret related to the corresponding session. The only information that the adversary can obtain by the physical attack is the identity of the device which is stored in the device's memory at the end of each authentication phase. However, knowing the device's ID will be trivially observable in the next session. 

\subsection{Performance Evaluation}
In this section, we compare our proposed protocol with the other PUF-based schemes in the literature \cite{24}-\cite{28} in terms of communication and computational overhead and security features guarantees.
We compute the communication overhead of our protocol by measuring the maximum size of messages transmitted from an \ac{iot} device to the server and vice versa. Thus, the total communication cost of our scheme is $|A_i^j,V1_i^j,Ch_i^j,HS_i^j |+|E_i^j,V2_i^j |= (4 \times 256)$ bit $+ (2 \times 256)$ bit $= 192$ Byte.
We assume that the \ac{iot} devices in the network are resource-constrained, whereas the server has a large database and high computational power. Hence, we consider only the computational overhead on the device's side. Nevertheless, our scheme leaves low computational overhead at the server side. As shown in section IV, EPUF needs $0.93$ ms to be executed on the mentioned hardware (with a maximum frequency of $100$ MHz) and generates errorless responses considering zero value as the threshold in \ref{alg:helpStream}. To simulate other cryptographic primitives on the device's side, we refer to \cite{27}, where a single core $798$ MHz CPU with 256 MB of RAM together with JCE library \cite{29} have been used. For calculating the execution time of a typical PUF operation and fuzzy extractor, we adopted a $128$ bit arbiter PUF and the code offset mechanisms using Bose, Chaudhuri, and Hocquenghem (BCH) \cite{30}. In our scheme, each \ac{iot} device needs four one-way hash functions, two PRNGs, and one EPUF operation for one authentication process, for a total of $1.134$ ms. 

Tab~\ref{tab:primitives} summarizes the primitives used.
\begin{table}[!h]
\renewcommand{\arraystretch}{1.3}
\caption{Primitive Functions Description and Computation Cost.}
\label{tab:primitives}
\centering
\begin{tabular}{|p{1cm}||p{5cm}|p{1.5cm}|}
\hline
\bfseries Primitive & \bfseries Description & \bfseries Computation Cost [ms] \\
\hline\hline
$T_{\rm EPUF}$ & EPUF operation & $0.93$ \\
\hline
$T_{\rm PUF}$ & 128-bit Arbiter PUF operation & $0.12$ \\
\hline
$ T_{\rm H}$ & SHA-3 function & $0.026$ \\
\hline
$T_{\rm RN}$ & 128-bit random number generation & $0.05$ \\
\hline
$T_{\rm HMAC}$ & HMAC time~\cite{26, 28} & $0.166$ \\
\hline
$T_{\rm SE} $ & AES-128 encryption & $0.158$ \\
\hline
$T_{\rm SD} $ & AES-128 decryption  & $0.254$\\
\hline
$T_{\rm PM}$ & point multiplication over elliptic curve & $5.9$ \\
\hline
$T_{\rm PA}$ & point addition over elliptic curve & $0.84$ \\
\hline
$T_{\rm ME}$ & modular exponential operation & $7.86$ \\
\hline
$T_{\rm REC}$ & fuzzy extractor key reconstruction & $3.28$ \\
\hline
$T_{\rm GEN}$ & fuzzy extractor key generation & $2.68$ \\
\hline
\end{tabular}
\end{table}
To compare our scheme with already available literature solutions, we consider the following criteria. $F_1$: data confidentiality and integrity; $F_2$: replay attack resistance; $F_3$: DoS attack resistance;$F_4$: impersonation attack resistance; $F_5$: privacy and anonymity; $F_6$: forward secrecy; $F_7$: backward secrecy; $F_8$: considering PUF reliability; $F_9$: proper error correction scheme for PUF; $F_{10}$: lightweight design. Tab.~\ref{tab:eval} shows the comparison results. Our protocol has the least computational overhead. Notice that the presented schemes in~\cite{26} and~\cite{28} did not considered the reliability problem of the PUF by assuming it ideal without any erroneous response. While, by considering a simple BCH-based fuzzy extractor as the most known error correction scheme for PUFs, their computational costs will increase to $6.16$ ms and $5.392$ ms, respectively. Hence, our EPUF-based protocol  is more than three times faster than the best scheme in the literature. One may argue that the registration phase in this paper puts lots of weight to the \ac{iot} devices. However, since our proposed EPUF can produce a very large set of \acp{crp}, the registration process is needed to be redone in very long time periods. Hence, the weight of this phase can be discarded. On the other hand, by registering the \ac{iot} devices in this way, not only we provided an ultra-lightweight authentication protocol but also very important security features are achieved as shown in Tab.~\ref{tab:eval}.

\section{Conclusion}
In this paper, we proposed EPUF, a novel approach to generate robust responses using the entropy features of latency-based \ac{dram} values. We implemented our proposed EPUF on a platform consisting of a Xilinx Spartan6 FPGA (XC6SLX45) and DDR3 chips. The implementation results showed that the proposed EPUF provides a large set of \acp{crp}, which can be used in key generation and authentication purposes without needing any ECC. In addition, we proposed a lightweight authentication protocol based on EPUF, which not only stands secure against CK-adversary but also outperforms the state-of-the-art in term of computational cost. 


\bibliographystyle{IEEEtran}


\end{document}